\RequirePackage{fix-cm}
\documentclass{svjour3}                     
\smartqed  
\usepackage{amsfonts,amsmath,amssymb}
\usepackage{indentfirst, setspace}
\usepackage{url,float}
\usepackage{color}
\usepackage{hyperref}
\usepackage{longtable}
\usepackage{graphicx}
\usepackage[a4paper,ignoreall]{geometry}
\usepackage{multicol}
\usepackage{stfloats}
\usepackage{enumerate}
\usepackage{cite}
\def\qu#1 {\fbox {\footnote {\ }}\ \footnotetext { From Qu: {\color{red}#1}}}
\def\hqu#1 {}
\def\yin#1 {\fbox {\footnote {\ }}\ \footnotetext { From Yin: {\color{blue}#1}}}
\def\hyin#1 {}

\newtheorem{Prob}[theorem]{Problem}
\newtheorem{Lemma}[theorem]{Lemma}
\newtheorem{result}[theorem]{Result}

\newcommand{\tr}{{\rm Tr}}
\newcommand{\rank}{{\rm rank}}

\newcommand{\gf}{{\mathbb F}}

\newcommand{\supp}{{\rm Supp}}

\makeatletter
\newcommand{\figcaption}{\def\@captype{figure}\caption}
\newcommand{\tabcaption}{\def\@captype{table}\caption}
\makeatother

\begin{document}

\title{More Constructions of Differentially 4-uniform Permutations on $\gf_{2^{2k}}$}


\author{Longjiang~Qu\and Yin~Tan\and Chao~Li \and Guang~Gong}

\authorrunning{L. Qu, Y. Tan, C. Li, G. Gong} 

\institute{L. Qu \at
           College of Science,
           National University of Defense Technology, ChangSha, 410073, China.
           \\
           \email{ljqu\_happy@hotmail.com}           
           \and
           Y. Tan and G. Gong \at
           Department of Electrical and Computer Engineering,
           University of Waterloo, Waterloo, Canada.
	   \\
	   \email{\{yin.tan, ggong\}@uwaterloo.ca}
	   \and
	   C. Li\at
           College of Science,
           National University of Defense Technology, ChangSha, 410073, China.
	   He is also with Science and Technology on Information Assurance Laboratory,
	   Beijing, 100072, China.
	   \\
	   \email{lichao\_nudt@sina.com}
}

\date{Received: date / Accepted: date}

\maketitle

\begin{abstract}
    Differentially $4$-uniform permutations on $\gf_{2^{2k}}$ with high nonlinearity
	are often chosen as Substitution boxes in both block and stream ciphers.
	Recently, Qu et al. introduced a class of functions,
	which are called preferred functions, to
	construct a lot of infinite families of such permutations \cite{QTTL}.
	In this paper, we propose a particular type of Boolean functions to characterize the preferred functions.
	On the one hand, such Boolean functions can be determined by solving linear equations,
	and they give rise to a huge number of differentially $4$-uniform permutations
	over $\gf_{2^{2k}}$. Hence  they may provide more choices for the design of Substitution boxes.
	On the other hand, by investigating the number of these Boolean functions,
	we show that the number of CCZ-inequivalent differentially $4$-uniform permutations over $\gf_{2^{2k}}$
	grows exponentially when $k$ increases, which gives a positive answer to an open
	problem proposed in \cite{QTTL}.
\end{abstract}

\section{Introduction}
In the design of many block ciphers and stream ciphers, permutations with {specific}
properties are chosen as  Substitution boxes to bring the confusion into ciphers.
To prevent various attacks to the cipher and for the software implementation,
such permutations are required to have low differential uniformity \cite{diff},
high algebraic degree \cite{degree}, high nonlinearity \cite{linear}, and being defined
on fields with even degrees, namely $\gf_{2^{2k}}$, {etc.}
Throughout this paper, we always let $n=2k$ be an even integer.

It is well known that the lowest differential uniformity of a function defined on $\gf_{2^n}$
can achieve is $2$ and such functions are called \textit{almost perfect nonlinear} (APN) functions.
However, due to the lack of  knowledge of the APN permutation in even dimension,
constructions of  differentially $4$-uniform permutations with
high algebraic degree and high nonlinearity over $\gf_{2^{2k}}$ have attracted many researchers' interest.
One may refer to \cite{Carl-Lea, carl-tan-tan, carlet-4uni, Carlet13, liWang, LWY, QTTL, QXL, TCT, ZHS}
for recent progress on
this problem. To the best of our knowledge, before \cite{QTTL}, not many infinite families
of differentially $4$-uniform permutations are known.
Among them, the multiplicative inverse function perhaps is the most popular one and it is
endorsed as the Substitution box {in the block ciphers like \textsf{AES} \cite{AES}, \textsf{Camellia} \cite{Camellia},
or as the filtering functions in the stream ciphers like \textsf{SFINK} \cite{SFINK}.}
In \cite{QTTL}, the authors used the so-called switching method ({introduced in \cite{Dillon,EP}})
to construct many infinite families of
such functions, which significantly increase the number of them (see \cite[Table 1]{QTTL}).
More precisely, they introduced a type of functions, which are called \textit{preferred functions},
to discover a lot of differentially $4$-uniform permutations on $\gf_{2^{2k}}$.
All obtained functions are with highest algebraic degree ($2k-1$) and {relatively high}
nonlinearity (see Result \ref{nl1} {and Table \ref{res23}}).

After obtaining many infinite families of differentially $4$-uniform permutations
over $\gf_{2^{2k}}$ as in \cite{QTTL}, determining the number of CCZ-inequivalent classes
among these functions arises as a natural question. It was observed in \cite{QTTL} that this number
grows exponentially when {$k$} grows, and they proposed  an open problem to determine,
or give a lower bound of this number \cite[Problem 4.16]{QTTL}.
Since more constructions of preferred functions lead to more
differentially $4$-uniform permutations ({see} Result \ref{res_general} below),
a characterization of preferred functions is
clearly helpful to solve the above problem. Hence it was proposed in \cite{QTTL} as
another open problem to find more or give a characterization of
preferred functions \cite[Problem 4.15]{QTTL}.

The purpose of this paper is to proceed further the  research of \cite{QTTL} and to
answer the above two problems.

We should mention that a  similar question about the number of CCZ-inequivalent classes of
APN functions has been proposed in \cite{EP}. Recently the results reported in \cite{tan-APN,yu-APN}
gave a positive evidence (thousands of APN functions are discovered
on $\gf_{2^7},\gf_{2^8}$ and $\gf_{2^9}$). However, a theoretical proof of
this problem currently seems unavailable.

In Section 2, we propose a new type of Boolean functions, called \textit{preferred Boolean functions}
{(Definition \ref{def_PFB}, PBF for short)},  to characterize preferred functions.
{Such Boolean functions are shown to be able to}
construct differentially $4$-uniform permutations (Corollary \ref{th_general}).
{More interestingly,} these Boolean functions can be efficiently determined. Precisely,
it is proven in Theorem \ref{ChPBF} that the determination of such Boolean functions can be reduced to
solving linear equations. By estimating the $2$-rank of the coefficient matrix of this linear
equation system, we prove that there exist at least $2^{\frac{2^n+2}{3}}$ preferred Boolean functions
with $n$ variables.
Hence,  by Corollary \ref{th_general}, we construct at least $2^{\frac{2^n+2}{3}}$  differentially $4$-uniform permutations
on $\gf_{2^n}$. This is a huge number, {which} shows that,
{by making use of preferred Boolean functions}, we may give much more differentially $4$-uniform
permutations than those given in \cite{QTTL}.
For example, we have $2^{86}$ preferred Boolean functions on $\gf_{2^8}$, and hence construct
the same number of differentially $4$-uniform permutations on $\gf_{2^8}$ (Corollary \ref{th_general}),
which is a great improvement of the $2^{\frac{n}{2}+7}=2^{11}$ such permutations constructed
in \cite{QTTL} (\cite[Remark 4.14]{QTTL}).

Based on the investigation of the number of preferred Boolean functions, in Theorem \ref{main1},
we show that the number of CCZ-inequivalent differentially $4$-uniform permutations
over $\gf_{2^n}$ {among those constructed in Corollary \ref{th_general}}, denoted by $N(n)$,
is at least $2^{\frac{2^n+2}{3}-4n^2-2n}$. This shows that
the number of differentially $4$-uniform permutations
on $\gf_{2^{2k}}$ \textbf{does} increase exponentially when $k$ grows.

In Table \ref{table-thm1} below, the value or a lower bound of $N(n)$ is listed for $6\leq n\leq 16$.
Some remarks on the table are as follows.
One may be curious about why $N(6)$ seems to be quite large while the lower bounds of $N(8)$ and $N(10)$
seems to be weak. The reasons are as follows. On the one hand,  thanks to Theorem \ref{ChPBF},
it can be shown that there are $2^{22}$ preferred Boolean functions on $\gf_{2^6}$. Hence
by Corollary \ref{th_general} we may construct $2^{22}$ differentially $4$-uniform permutations on $\gf_{2^6}$.
The small number of PBFs  on this field enables us to perform an exhaustive search
of CCZ-inequivlaent such functions on $\gf_{2^6}$, namely to determine the exact value of $N(6)$.
The method we distinguish two CCZ-inequivalent
functions is to compare whether the codes corresponding to them are equivalent
(see {\cite{CCZ} or \cite[page 43]{carlet-book}}).
On the other hand,  when $n=8,10$, the authors in \cite{QTTL} only counted the number of differentially $4$-uniform
permutations with different differential spectrum, which is only an invariant of CCZ-equivalence.

\begin{center}
   \tabcaption{The number of CCZ-inequivalent differentially $4$-uniform
	   permutations among those \newline constructed in Corollary \ref{th_general} on $\gf_{2^n}$
when $6\leq n\leq 16$ ($n$ even) }
       \begin{tabular}{|c|c|c|c|} \hline
       \label{table-thm1}
        $n$    & $ N(n) $   & Ref                          \\ \hline
	$6$    & $11,120$   & Exhaustive search       \\ \hline
	$8$    & ${\geq} 107$      & \cite[Table 1]{QTTL}         \\ \hline
	$10$   & ${\geq} 183$      & \cite[Table 1]{QTTL}         \\ \hline
	$12$   & ${\geq} 2^{766}$      & Theorem \ref{main1}(iii)         \\ \hline
	$14$   & ${\geq} 2^{4650}$    & Theorem \ref{main1}(iii)  \\ \hline
	$16$   & ${\geq} 2^{20790}$   & Theorem \ref{main1}(iii)           \\ \hline
 \end{tabular}
\end{center}

In Section 2.3, we relate the lower bound of $N(n)$
to the rank of the coefficient matrix $M$ of
certain linear equation system  (Problem \ref{prob_M}), and the problem to determine a lower bound of
the number of the functions which are CCZ-equivalent to a given function (Problem \ref{prob_CCZNum}).
For Problem \ref{prob_M}, in Section 2.4, we provide a possible method {to solve it}
by relating the rank of the coefficient
matrix to the existence of a $3$-regular subgraph
of the graph determined by the matrix (Definition \ref{def_G}). Some properties of this graph
are also presented therein.

As known from Theorem \ref{main1}(i), the set of all preferred Boolean functions is a $\gf_2$-subspace.
In Section 3, we investigate the so-called \textit{non-decomposable} preferred Boolean
functions. A preferred Boolean function $f$ is called \textit{non-decomposable} if it is not the sum
of the other two preferred Boolean functions whose support sets are
 proper subsets of the support of $f$. In Theorem \ref{NonDecSolution}, we present the characterization
of these functions in terms of their support sets. Then many
explicit constructions of differentially $4$-uniform permutations are presented.
It is an evidence that preferred Boolean functions is a more efficient tool than preferred functions
to construct differentially $4$-uniform permutations.

The rest of the paper is organized as follows. In Section 2, we first recall some results {appeared}
in \cite{QTTL}, and then introduce the definition and the properties of the preferred Boolean functions
and its application to construct differentially $4$-uniform permutations, and
followed by presenting a lower bound of the number of CCZ-inequivalent classes of constructed permutations.
In Section 2, we focus on the characterization of
non-decomposable preferred Boolean functions. Some concluding remarks are given in Section 4.

We end this section by introducing some notations.
Given two positive integers $n$ and $m$, a function
$F:\gf_{2^n}\rightarrow\gf_{2^m}$ is called an
\textit{$(n,m)$-function}. Particularly, when $m=1$, $F$ is called an $n$-variable
\textit{Boolean function}, or a \textit{Boolean function} with $n$ variables.
Clearly, a Boolean function may be regarded as a vector
with elements on $\gf_2$ of length $2^n$ by identifying $\gf_{2^n}$
with a vector space $\gf_2^n$ of dimension $n$ over $\gf_2$.
In the following, we will switch between these two points of view without
explanation if the context is clear. Let $f$ be a nonzero Boolean function.
Define the set $\supp(f)=\{x\in \gf_{2^n}| f(x)=1\}$ and call it the \textit{support set} of $f$.
The value $|\supp(f)|$ is called the \textit{(Hamming) weight} of $f$.
Denote by $\tr(x)=\sum_{i=0}^{n-1}x^{2^i}$ the absolute trace function
from $\gf_{2^n}$ to $\gf_{2}$.
Note that for the multiplicative inverse function
$x^{-1}$, we define $0^{-1} = 0$ as usual.

\section{Preferred Functions and Preferred Boolean Functions}
In this section, we
introduce preferred Boolean functions to
characterize preferred functions. Then the properties of preferred Boolean functions
are investigated. Finally, a lower bound of the number of CCZ-inequivalent differentially
$4$-uniform permutations on $\gf_{2^{2k}}$ is presented.

\subsection{Preferred functions}

First, we recall the definition of preferred functions.
\begin{definition} \cite{QTTL}
	\label{def_PF}
	Let $n=2k$ be an even integer and $R$ be an $(n,n)$-function.
	Define the Boolean function $D_R$ by
        \begin{eqnarray}
		\label{eq_T0}
                 D_R(x)=\tr (R(x+1)+R(x)),
        \end{eqnarray}
        and define the functions $Q_R, P_R$ by
        \begin{eqnarray}
           \label{eq_Q0}
	       Q_R(x,y) &=& D_R\left(\frac{1}{x}\right)+D_R\left(\frac{1}{x}+y\right),\ \mbox{and}\\
           \label{eq_P0}
	       P_R(y) &=& Q_R(0,y)=D_R(0)+D_R(y).
        \end{eqnarray}
        Let $U$  be the subset of $\gf_{2^n}\times\gf_{2^n}$ defined by

        \begin{eqnarray}
	  \label{U}
	     U = \{(x,y)|x^2+\frac{1}{y}x+\frac{1}{y(y+1)}=0, y\not\in\gf_2\}.
        \end{eqnarray}
        If
        \begin{eqnarray}
            \label{eq_Preferred}
               Q_R(x, y)+P_R(y)=0
        \end{eqnarray}
        holds for any pair $(x,y)$ in $U$, then we call $R$ a
        \textit{preferred function} (PF for short) on $\gf_{2^n}$, or call it
        \textit{preferred} on $\gf_{2^n}$. In particular, if $D_R$ is the
        zero function, then clearly $R$ is a PF and it is called \textit{trivial}.
\end{definition}

Preferred functions may be used to construct differentially $4$-uniform permutations over $\gf_{2^{2k}}$
as follows.

\begin{result}
     \cite[Theorem 3.6]{QTTL}
     \label{res_general}
      Let $n=2k$ be an even integer, $I(x)=x^{-1}$ be the multiplicative inverse function and $R$ be an
      $(n,n)$-function. Define
      \begin{eqnarray*}
	 && H(x)=x + \tr(R(x)+R(x+1)),\ \mbox{and} \\
	 && G(x)=H(I(x)).
      \end{eqnarray*}
      If $R(x)$ is a preferred function, then $G(x)$ is a differentially
      $4$-uniform permutation polynomial on $\gf_{2^n}$.
\end{result}

A lower bound of the nonlinearity of the functions defined in the above result is as follows.
\begin{result}
     \cite[Theorem 5.4]{QTTL}
	\label{nl1}
	For any positive even integer $n$, let $G$ be a function as
	defined in Result \ref{res_general}.
	Then we have
	$\mbox{NL}(G)\geq 2^{n-2}-\frac{1}{4}\lfloor2^{\frac{n}{2}+1}\rfloor-1,$
{where $\mbox{NL}(G)$ denotes the nonlinearity of the function $G$}.
\end{result}

\subsection{Preferred Boolean functions}

Now we introduce a new type of Boolean functions, called \textit{preferred Boolean function} below,
to characterize preferred functions.

\begin{definition}
	\label{def_PFB}
        Let $n=2k$ be an even integer and $f$ be an $n$-variable Boolean function.
        We call $f$ a \textit{preferred Boolean function} (PBF for short)
        if it satisfies the following two conditions:
        \begin{enumerate}[(i)]
           \item $f(x+1)=f(x)$ for any $x\in \gf_{2^n}$;
           \item $f\left(\frac{1}{x}\right)+f\left(\frac{1}{x}+y\right)+f(0)+f(y)=0$
                   for any pair $(x,y)\in U$, where $U$ is defined by \eqref{U} in Definition \ref{def_PF}.
        \end{enumerate}
\end{definition}

The following result points out the relationship between preferred functions and preferred Boolean functions.
\begin{proposition}
        \label{PFBrela}
	Let $n=2k$ be an even integer and $R$ be an $(n,n)$-function. Define the Boolean
	function $D_R$ as in (\ref{eq_T0}).
	Then $R$ is a preferred function  if and only if $D_R$ is a preferred Boolean function.
        Furthermore, for any preferred Boolean function $f$ with $n$ variables, there are
	$2^{n\cdot 2^n-2^{n-1}}$ preferred functions $R$ such that $D_R(x)=f(x)$.
	In particular, there are  $2^{n\cdot 2^n-2^{n-1}}$
	trivial preferred functions on $\gf_{2^n}$.
\end{proposition}
\begin{proof}
        The first part follows directly from the definitions of PF and PBF.
	Now, given a PBF $f$ on $\gf_{2^n}$, suppose that $R$ is an $(n,n)$-function
	satisfying  $D_R(x)=\tr(R(x+1)+R(x))=f(x)$.
	Next we count the number of such $R$.
	Note that there are $2^{n-1}$ pairs of $(x,x+1)$ in $\gf_{2^n}$. For each such pair,
	$R(x)$ can be any value in $\gf_{2^n}$, and $R(x+1)$ can
	be any value such that   $\tr(R(x+1)+R(x))=f(x)$ is fixed, which means that
	$R(x+1)$ can take $2^{n-1}$ different values. Thus we have
	in total  $\left( 2^{n} \cdot 2^{n-1}\right)^{2^{n-1}}=2^{n\cdot 2^n-2^{n-1}}$
	preferred functions $R$ such that $D_R(x)=f(x)$. The statement about the number of
	trivial preferred functions then follows. We complete the proof.
$\hfill\Box$\end{proof}

By Result \ref{res_general} and Proposition \ref{PFBrela}, it is clear
that PBFs can be used to construct differentially
$4$-uniform permutations.

\begin{corollary}
	\label{th_general}
	Let $n=2k$ be an even integer, $I(x)=x^{-1}$ be the
	multiplicative inverse function and $f$ be a {Boolean function with $n$ variables.}
	Define
	\begin{eqnarray*}
		&& H(x)=x + f(x),\ \mbox{and} \\
		&& G(x)=H(I(x)).
	\end{eqnarray*}
	If $f(x)$ is a preferred Boolean function, then $G(x)$ is a
	differentially $4$-uniform permutation polynomial {on $\gf_{2^n}$}.
\end{corollary}

\subsection{Number of CCZ-inequivalent differentially $4$-uniform permutations}

It is clear that $f$ is a preferred Boolean function if and only if so is $f+1$.
For convenience, we assume that $f(0)=0$ in the rest of the paper.
The following result provides a method to tell whether a Boolean function is
a PBF.

\begin{theorem}
	\label{ChPBF}
	Let $n=2k$ be an even integer and $f$ be an $n$-variable Boolean function.
	Let $\omega$ be an element of $\gf_{2^n}$ with order 3. Then $f$
	is a preferred Boolean function if and only if it satisfies the following two {conditions}:
	\begin{enumerate}[(i)]
	   \item $f(x+1)=f(x)$ for any $x\in \gf_{2^n}$;
	   \item $f\left(\alpha+\frac{1}{\alpha}\right)+f\left(\omega\alpha+\frac{1}{\omega\alpha}\right)
		    +f\left(\omega^2\alpha+\frac{1}{\omega^2\alpha}\right)=0$
	            for any $\alpha\in \gf_{2^n}\setminus \gf_{4}$.
	\end{enumerate}
\end{theorem}

\begin{proof}
	We first determine the elements in $U$ ({recall its definition in (\ref{U})}).
	Given any $(x,y)\in U$, we have
	$x^2+\frac{1}{y}x+\frac{1}{y(y+1)}=0$. {It is easy to see that} this equation has solutions in $\gf_{2^n}$ if and
	only if
	$\tr\left(\frac{1}{y(y+1)}\cdot y^2\right)=\tr\left(\frac{y}{1+y}\right)=\tr(\frac{1}{1+y})=0$.
	Furthermore, by \cite[Lemma 4.1]{lachaud},
	for such $y$, there exists $\alpha\in\gf_{2^n}^*$ such that $y+1=\alpha+\alpha^{-1}$.

	Next we solve the equation
	\begin{equation}
		\label{eqU}
		x^2+\frac{1}yx+\frac{1}{y(y+1)}=0.
	\end{equation}
	{Recall that the existence of $\omega\in\gf_{2^n}$ with order 3 is guaranteed by $3\mid 2^{2k}-1$
	and such element satisfies $1+w+w^2=0$.
	Let $z\in \gf_{2^n}$ such that} $x=\frac{1}{y}(z+\omega)$. Then we have
	$z^2+z+\frac{1}{y+1}=0$. With $y+1=\alpha+\alpha^{-1}$, we get
	$z=\frac{1}{\alpha+1}$ or $z=\frac{\alpha}{\alpha+1}$.
        If $z=\frac{1}{\alpha+1}$, then
\begin{eqnarray*}
  \frac{1}{x} &=& \frac{y}{z+w}=\frac{\alpha+\alpha^{-1}+1}{(\alpha+1)^{-1}+w}=\frac{(\alpha+1)^2\alpha+(\alpha+1)}{\alpha(w(\alpha+1)+1)} \\
   &=& \frac{\alpha^3+1}{w\alpha(\alpha+w)}= \frac{\alpha^2+w\alpha +w^2}{w\alpha } = \omega^2\alpha+\frac{1}{\omega^2\alpha} +1.
\end{eqnarray*}
Hence
	\begin{equation*}
	  \left( \frac{1}{x}+1, \frac{1}{x}+y \right)=
	  \left(\omega^2\alpha+\frac{1}{\omega^2\alpha},
	        \omega\alpha+\frac{1}{\omega\alpha} \right).
	\end{equation*}
	If $z=\frac{\alpha}{1+\alpha}$, similar
	computations give the following result:
	\begin{equation*}
	  \left( \frac{1}{x}+1, \frac{1}{x}+y \right)
	  =\left( \omega\alpha+\frac{1}{\omega\alpha}, \omega^2\alpha+\frac{1}{\omega^2\alpha} \right).
	\end{equation*}
	The rest of the proof follows from the definition of PBF, the assumption $f(0)=0$ and
	the fact that $y\in \gf_2$ if and only if $\alpha \in \gf_4$.
	We complete the proof.
$\hfill\Box$\end{proof}

The characterization of PBFs in Theorem \ref{ChPBF} is useful to determine all PBFs.
First note that, when $\alpha \notin \gf_4$, the elements
$\alpha+\frac{1}{\alpha}, \omega\alpha+\frac{1}{\omega\alpha}, \omega^2\alpha+\frac{1}{\omega^2\alpha}$ are
all distinct (since the sum of them is $0$, and none of them can be zero).
In the following result, we denote the set of all PFs and PBFs by $\mathcal{PF}$ and $\mathcal{PBF}$,
respectively.

\begin{theorem}
  \label{main1}
	Let $n=2k$ be an even integer. Then the following results hold:
	\begin{enumerate}[(i)]
	  \item The set $\mathcal{PF}$ (resp. $\mathcal{PBF}$) are $\gf_2$-subspaces of the set of
		  all $(n, n)$-functions (resp. the set of all $n$-variable Boolean functions).

	   \item
	   {
		Define the following two sets:
		\begin{eqnarray*}
			L_1 &=& \Big\{ \{x, x + 1\} : x\in\gf_{2^n}\setminus\gf_2  \Big\}, \\
			L_2 &=& \left\{ \{ \alpha+\frac{1}{\alpha}, \omega\alpha+\frac{1}{\omega\alpha},
					\omega^2\alpha+\frac{1}{\omega^2\alpha} \} :
					\alpha\in \gf_{2^n}\setminus\gf_4 \right\}.
		\end{eqnarray*}
		Let $v_{x}$ and $v_{\alpha}$ be the characteristic function
		in $\gf_{2^n}\setminus\gf_2$ of each
		$\{ x, x+1 \} \in L_1 $ and
		$\Big\{ \alpha+\frac{1}{\alpha}, \omega\alpha+\frac{1}{\omega\alpha},
		 \omega^2\alpha+\frac{1}{\omega^2\alpha} \Big\}\in L_2$,
		respectively.
		Define the $(|L_1|+|L_2|)\times (2^n-2)$ matrix $M$ by
		\begin{equation}
		  \label{M}
			M = \left[
				\begin{array}{lll}
					  & v_x &  \\
					  & v_{\alpha} &
				  \end{array}
			  \right],
		\end{equation}
		where the columns and rows of $M$ are indexed by the elements in $\gf_{2^n}\setminus \gf_2$ and
		$L_1\cup L_2$ respectively.}
		Then the dimension of $\mathcal{PBF}$ is $2^n-1-\rank(M)$, and the dimension of
		$\mathcal{PF}$ is $ n\cdot 2^{n}+2^{n-1}-1-\rank(M)$.
	  \item There are at least $2^{\frac{2^{n}+2}{3}-4n^2-2n}$ CCZ-inequivalent
		  differentially $4$-uniform permutations over $\gf_{2^n}$
		  among all the functions constructed by Corollary \ref{th_general}
		  or equivalently, by Result \ref{res_general}.
	\end{enumerate}
\end{theorem}

\begin{proof}
        (i) It is clear from the definitions  that $\mathcal{PBF}$ and $\mathcal{PF}$ are $\gf_2$-subspaces.

	(ii) First note that determining $f$ is equivalent to determining all the images $f(x)$ for $x\in\gf_{2^n}$.
	Now let $f$ be a PBF (recall that we assume $f(0)=f(1)=0$, and therefore we only need to determine the images of
	$f$ on $\gf_{2^n}\setminus\gf_2$). By the two conditions in Theorem \ref{ChPBF}, clearly
	we may obtain all such PBFs by solving linear equations as the following.

	It is not difficult to see that $|L_1|=2^{n-1}-1$. Note that when $\alpha\in \gf_{2^n}\setminus\gf_4$,
	the six elements $\alpha, \omega\alpha, \omega^2\alpha , \frac{1}{\alpha},
	\frac{1}{\omega\alpha}, \frac{1}{\omega^2\alpha}$ are distinct, and any one
	leads to the same element of  $L_2$.
	Hence $|L_2|=\frac{2^n-4}{3\cdot 2}=\frac{2^{n-1}-2}{3}$.
	It follows from Theorem \ref{ChPBF} that  a Boolean function
	$f$ is a PBF with $f(0)=0$ if and only if
	$$ M f^{\rm{T}}=0. $$
        Note that in the above equation, by abuse of notation,
        we still use $f^{\rm{T}}$ to denote the value vector of $f$ on $\gf_{2^n}\setminus \gf_2$.
        Therefore the dimension of the set of all PBFs with $f(0)=0$ is  $2^n-2-\rank(M)$.
	
	It is clear that $f+1$ is also a PBF if $f$ is a PBF, hence altogether the
	dimension of $\mathcal{PBF}$ is $2^n-2-\rank(M)+1=2^n-1-\rank(M)$.
	The  dimension of $\mathcal{PF}$ then
        follows from that of $\mathcal{PBF}$ and Proposition \ref{PFBrela}.
	
	(iii) On the one hand, we have
	\begin{eqnarray*}
		\rank(M) & \leq & \min \{|L_1|+|L_2|, 2^n-2\}  \\
		         &=&  \min\{ \frac{2^{n+1}-5}{3}, 2^n-2 \} =   \frac{2^{n+1}-5}{3}.
	\end{eqnarray*}
	Hence, the dimension of $\mathcal{PBF}$, which is one plus the dimension of the null space of $M$, is at least
	$2^n-2-\frac{2^{n+1}-5}{3}+1=\frac{2^{n}+2}{3}$.
	It then follows from Corollary \ref{th_general} that we can obtain at
	least $2^{\frac{2^{n}+2}{3}}$ differentially $4$-uniform permutations on $\gf_{2^n}$.
	
	On the other hand, for any $(n,n)$-function,
	there are at most
        $(2^n)^{4n+2}=2^{4n^2+2n}$ functions which are CCZ-equivalent to it.
        Indeed, given two $(n,n)$-functions $F$ and $G$. If $G$ is CCZ-equivalent to $F$,
	then by definition there exists an affine automorphism $L$ of $\gf_{2^n}\times\gf_{2^n}$
	such that
	\begin{equation}
	  \label{CCZeq1}
	  L\Big(\{(x, F(x)):x\in\gf_{2^n}\}\Big)=\Big\{(y, G(y)):y\in\gf_{2^n}\Big\}.
	\end{equation}
	Clearly we may write $L=(L_1,L_2)$, where $L_i$ is an affine function from
	$\gf_{2^n}\times\gf_{2^n}$ to $\gf_{2^n}$. Define the functions
	$\mathcal{L}_i(x)=L_i(x, F(x))$ for $i=1,2$, we may rewrite
	the equation (\ref{CCZeq1}) as
	\begin{eqnarray*}
	   \Big\{(y, G(y)):y\in\gf_{2^n}\Big\}
	   &=& L\Big(\{(x, F(x)):x\in\gf_{2^n}\}\Big) \\
	   &=& \Big\{(L_1(x,F(x)), L_2(x, F(x))) : x\in\gf_{2^n}\Big\} \\
	   &=& \Big\{(\mathcal{L}_1(x), \mathcal{L}_2(x)) : x\in\gf_{2^n}\Big\} \\
	   &=& \Big\{(x, \mathcal{L}_2(\mathcal{L}_1^{-1}(x)) : x\in\gf_{2^n}\Big\}.
        \end{eqnarray*}
	which implies that $G=\mathcal{L}_2(\mathcal{L}_1^{-1}(x))$, and then we can see that
	once $\mathcal{L}_1$ and  $\mathcal{L}_2$ are determined, $G$ is determined.
	Clearly an affine mapping $\varphi:\gf_{2^n}\times\gf_{2^n}\rightarrow\gf_{2^n}$
	can be represented as the form
	$$ \varphi(x,y)=\sum_{i=0}^{n-1}a_ix^{2^i}+\sum_{i=0}^{n-1}b_iy^{2^i}+c, $$
	where $a_i, b_i, c\in \gf_{2^n}$.
        Therefore, there are at most $(2^n)^{2n+1}$ choices of $L_1$ and $L_2$ respectively, and then altogether
	there are at most $(2^n)^{4n+2}$ functions $G$ which are CCZ-equivalent to $F$.

	Finally, it follows
	that the number of CCZ-inequivalent differentially $4$-uniform permutations
	on $\gf_{2^n}$ is at least
	$$
	\frac{2^{\frac{2^{n}+2}{3}}}{2^{n(4n+2)}}=2^{\frac{2^{n}+2}{3}-4n^2-2n}.
	$$
        We complete the proof.
$\hfill\Box$\end{proof}

Several remarks on the above theorem are in the sequel.

\begin{remark}
\label{remark_prob}
(i) By MAGMA, when $6\leq n\leq 14$, we compute the rank of the matrix $M$ defined in (\ref{M})
    and found that they are all equal to $\frac{2^{n+1}-5}{3}$. We cannot prove this and leave it
    to interested readers, and we will present some possible methods to solve this problem
    in the next subsection.

(ii) In the proof of Theorem \ref{main1}(iii), the bound on the number of the
     functions $G$ which are CCZ-equivalent
     to a given function $F$ is  loose as we do not consider the requirement that
     $L$ is a permutation. {An improvement on} such number would yield a better lower bound of the number of
     CCZ-inequivalent differentially $4$-uniform permutations. We propose it as an open problem.

(iii) Theorem \ref{main1}(iii) shows that the number of
      CCZ-inequivalent differentially $4$-uniform permutations over $\gf_{2^n}$ ($n$ even)
      grows double exponentially when $n$ grows.
      This has been observed in \cite{QTTL} and proposed as an open problem.
      To the best of our knowledge,
      this is the first bound on the number of CCZ-inequivalent classes of
      differentially $4$-uniform permutations on $\gf_{2^n}$ with $n$ even.
\end{remark}

{According to Remark \ref{remark_prob} (i) and (ii), we express the following two problems which
remain to be solved.}
\begin{Prob}\label{prob_M}
	{To prove that the rank of the matrix $M$ defined in (\ref{M}) is
              $\frac{2^{n+1}-5}{3}$.}
\end{Prob}

\begin{Prob}\label{prob_CCZNum}
	Given an $(n,n)$-function (or a permutation) $F$, determine the number, or
	an upper bound of the number of the
	functions $G$ which are CCZ-equivalent to {$F$}.
\end{Prob}

In the end of this subsection, we {give some comments} on
the nonlinearity of the differentially $4$-uniform permutations over $\gf_{2^n}$
constructed by Corollary \ref{th_general}.
Clearly, the lower bound of the nonlinearity in Result \ref{nl1} holds for all these functions. However,
as mentioned in \cite[page 11]{QTTL}, this bound is not tight. To show this, on small fields,
we randomly choose some differentially $4$-uniform permutations
constructed in Corollary \ref{th_general} and compute their nonlinearity.
Thanks to the characterization of PBF by solving linear equation system,
we can perform such a random search.
The computational results are given in the following table. The sample size means how many
differentially $4$-uniform permutations we choose, Ave(NL) denotes the average nonlinearity
of them, and KMNL denotes the known maximal value of nonlinearity $2^{n-1}-2^{n/2}$.
Var(NL) denotes the variance of the nonlinearity defined by
$\mbox{Val(NL)}=\sqrt{\sum\limits_{\mbox{nl}\in\mathcal{SP}}
(\mbox{nl}-\mbox{Ave(NL)})^2/|\mathcal{SP}|}$, where
$\mathcal{SP}$ is the chosen sample space. Finally, $\mbox{Dist}({\mbox{NL}})$ denotes the
distribution of the nonlinearity for functions in the sample space,
where $a^b$ means that there are $b$ sample
functions with nonlinearity $a$.

\begin{center}
   \tabcaption{Nonlinearity of the differentially $4$-uniform permutations
	   constructed by Corollary \ref{th_general} \newline
	   on $\gf_{2^n}$ when $6\leq n\leq 10$ ($n$ even)}
       \begin{tabular}{|c|c|c|c|l|c|c|} \hline
       \label{res23}
       $n$    & Sample size   & Ave(NL)   & Var(NL) & Dist(NL) & $\begin{array}{ll}
                                                        & \mbox{Bound in}\\
							& \mbox{Result \ref{nl1}}
						      \end{array}$   & $\begin{array}{ll}
						                         & \mbox{KMNL}
									\end{array}$\\ \hline \hline
	$6$    & $10,000$      & $18.4022$  &$1.2034$ &$\begin{array}{ll}
	                                         &14^{48}, 16^{849}, 18^{6161}\\
						 &20^{2928}, 22^{14}
						\end{array}$ & $14$  & $24$     \\ \hline
	$8$    & $10,000$      & $94.2740$  &$2.2576$ &$\begin{array}{ll}
	                                      &82^{10}, 84^{30}, 86^{30}, 88^{150} \\
					      &90^{540}, 92^{1620}, 94^{3450} \\
					      &96^{3490}, 98^{680} \\
					      \end{array}$
	& $55$  & $112$           \\ \hline
	$10$   & $5,000$       & $434.2524$  &$3.7225$ & $\begin{array}{ll}
	                                                  &418^4, 420^{16}, 422^5,424^{35} \\
							  &426^{132}, 428^{263}, 430^{470} \\
							  &432^{730}, 434^{1053}, 436^{1022} \\
							  &438^{910},440^{315}, 442^{45} \\
	                                                  \end{array}$ & $239$  & $480$      \\ \hline
 \end{tabular}
\end{center}

{Two remarks on Table \ref{res23} are as follows:
\begin{enumerate}[(1).]
	\item When $n=6$, by Theorem \ref{main1}, there are $2^{2^6-1-\rank(M)}$
	PBFs. By MAGMA, we compute that $\rank(M)=41$ and therefore altogether
	we obtain $2^{22}$ PBFs.
	We exhaustively search these $2^{22}$ PBFs and find that
	all associated differentially 4-uniform permutations with the known maximal
	nonlinearity (i.e. $2^5-2^3=24$) are CCZ-equivalent to the inverse function.
	The second highest nonlinearity of the functions can achieve is $22$.
	This is not suprising as such functions can be obtained via PBF of weight $2$
	(see \cite[Corollary 1]{yu-4uni} for instance).
        \item When $n=8,10$, by Theorem \ref{main1} and MAGMA, there are $2^{86}$ and $2^{342}$ PBFs.
	It is impossible for an exhaustive search to explore the nonlinearity
	property of the functions in Theorem \ref{main1}.
	Hence we do a random search. 
	The computational results in Table 2 suggest that on average the nonlinearity of the
	differentially $4$-uniform permutations 
	constructed by Theorem \ref{main1} approximates to the asympotic nonlinearity 
	of random $(n,n)$ functions, which is
	in the neighborhood of $2^{n-1}-2^{\frac{n}{2}}\sqrt{n \ln2}$ (\cite[Theorem 1]{Dib}). 
	There may exist differentially 4-uniform permutations with
	known maximal nonlinearity  but not CCZ-equivalent to the inverse function.
	However, the density of them is too small and we do not know how to find them efficiently.
	We leave the problem that whether there exist differentially 4-uniform permutations
	with known maximal nonlinearity for functions in Theorem \ref{main1} to interested
	readers. {To solve this problem, in addition to the experiment performed
	as in Table 2, we did the following test. When $n=8$, we exhaust all PBFs with
	Hamming weights $2,4,6$ (the reason we pick up low Hamming weight PBFs is that
	they have a higher chance to give rise to differentially 4-uniform permutations
	with known maximal nonlinearity) and compute their nonlinearity. Unfortunately
	we do not find a desirable example.}
\end{enumerate}

\subsection{More on the Rank of $M$}

In this subsection, we present a possible method to determine the rank of $M$ using graph theory.
For convenience, we first introduce some notations.
For any $\alpha\in \gf_{2^n}\setminus\gf_4$, we call the set
$A_{\alpha}=\{\alpha+\frac{1}{\alpha}, \omega\alpha+\frac{1}{\omega\alpha},
\omega^2\alpha+\frac{1}{\omega^2\alpha}\}$
a \textit{triple set} with respect to $\alpha$ (or TS for short).
Define the set $\mathcal{TS}=\{ A_{\alpha}|\alpha\in \gf_{2^n}\setminus\gf_4\}$.
It is known from the last subsection that $|\mathcal{TS}|=\frac{2^{n-1}-2}{3}$.

We first present some properties of triple sets defined above.

\begin{Lemma}
	\label{SUnion}
	The set $\mathcal{TS}$ is a partition of the set
	$S=\{ x\in \gf_{2^n}\setminus\gf_2 | \tr\left(1/x\right)=0\}$.
\end{Lemma}

\begin{proof}
	Clearly $\bigcup_{\alpha\in \gf_{2^n}\setminus\gf_4} A_{\alpha}\subseteq S$.
	Conversely, for any $y\in S$, there exists $\alpha\in \gf_{2^n}\setminus\gf_4$
	such that $y=\alpha+\frac{1}{\alpha}$.
	It follows that $y\in A_{\alpha}$. Therefore,
	$S \subseteq \bigcup_{\alpha\in \gf_{2^n}\setminus\gf_4} A_{\alpha}$.
	So $S=\bigcup_{\alpha\in \gf_{2^n}\setminus\gf_4} A_{\alpha}$.

	Furthermore, assume that $a\in A_{\alpha} \cap A_{\beta}$, where $A_{\alpha}, A_{\beta}\in\mathcal{TS}$.
	Without loss of generality, assume that $a=\alpha+\frac{1}{\alpha}=\beta+\frac{1}{\beta}$.
	We have either $\alpha=\frac{1}{\beta}$ or $\alpha=\beta$
	{as $x+1/x$ is a 2-to-1 mapping from $\gf_{2^n}\setminus\gf_4$ to $S$}. It then follows that
	in {each} case, we have $A_{\alpha}=A_{\beta}$. The proof is finished.
$\hfill\Box$\end{proof}

\begin{Lemma}
	\label{TS}
	Let $A=\{a_1, a_2, a_3\}$ {be a triple set}. Then the following results hold:
	\begin{enumerate}[(i)]
	  \item $a_1 + a_2 + a_3 = 0$;
	  \item $1+a_i\notin A$ for any $1\leq i\leq 3$.
	\end{enumerate}
\end{Lemma}

\begin{proof}
	(i) follows directly from the definition of the triple set.
	(ii) can be deduced from (i) and the fact that $1\notin A$.
$\hfill\Box$\end{proof}

We define the following sets for later usage:
\begin{eqnarray}\label{eqT}
	T_1 &=& \{x\in \gf_{2^n}| \tr\left(\frac{1}{x}\right)=\tr\left(\frac{1}{x+1}\right)=1\},\nonumber  \\
	T_2 &=& \{x\in \gf_{2^n}| \tr\left(\frac{1}{x}\right)+\tr\left(\frac{1}{x+1}\right)=1\},  \\
	T_3 &=& \{x\in \gf_{2^n}| \tr\left(\frac{1}{x}\right)=\tr\left(\frac{1}{x+1}\right)=0\}.\nonumber
\end{eqnarray}

{The following lemma determines the size of the sets $T_1,T_2,T_3$.
We will use this result in the next section.
\begin{lemma}
Let $n$ be an even positive integer and the sets $T_i$ are defined in (\ref{eqT}) for $i=1,2,3$.
Then
\label{sizeTi}
\begin{eqnarray*}
|T_1| &=& |T_3| = \frac{1}{4}\left( 2^n + 1 - 2^{n/2+1}\cos\left(n\arccos(1/\sqrt{8})\right)\right), \ \mbox{and}  \\
|T_2| &=& \frac{1}{2}\left(2^n - 1 + 2^{n/2+1}\cos\left(n\arccos(1/\sqrt{8})\right)\right).
\end{eqnarray*}
\end{lemma}
\begin{proof}
Since $n$ is even, clearly $0\in T_3$ as $\tr(1)=0$.
Let $y=1/x$, we have $\tr(1/x)=\tr(y)$ and $\tr(1/(x+1))=\tr(y/(y+1))=\tr(1/(y+1))$.
Furthermore, let $z=y+1$, we have $\tr(y)=\tr(z)$ and $\tr(1/(y+1))=\tr(1/z)$.
Therefore we may rewrite the sets $T_i$ in (\ref{eqT}) as follows:
\begin{eqnarray}
     \label{neweqT}
	T_1 &=& \{ z\in \gf_{2^n}| \tr(z) = \tr(1/z) = 1\},\nonumber  \\
	T_2 &=& \{ z\in \gf_{2^n}| \tr(z)+ \tr(1/z)=1\},  \\
	T_3 &=& \{ z\in \gf_{2^n}| \tr(z) = \tr(1/z)=0\}.\nonumber
\end{eqnarray}
Now the results we want to show are followed from \cite[page 9]{Hirschfeld}
and the fact that $\sum_{i=1}^3|T_i|=2^n$.
We finish the proof.
$\hfill\Box$
\end{proof}

}

It is clear that $T_1, T_2$ and $T_3$ is a partition of $\gf_{2^n}$
and for any element $a$ of $S$ {(defined in Lemma \ref{SUnion})}, either $a\in T_2$ or $a\in T_3$.

\begin{definition}
	\label{adj}
	Let $A_1$ and $A_2$ be two {triple sets. They are called}
	\textit{adjacent} if there exist $a\in A_1$ and $b\in A_2$ such that $a+b=1$.
	{To be more clear}, we call $A_2$ is adjacent to $A_1$ at $a$, and
	call $A_1$ is adjacent to $A_2$ at $b$.
	{Two adjacent triple sets are also called \textit{neighbours} in the following.}
\end{definition}

It follows from Lemma \ref{TS} that any TS can not be a neighbour of itself.
It is also clear that any TS has at most three neighbours.

\begin{Lemma}
	\label{Neigh}
	Let $A$ be a triple set. Then $A$ has either three neighbours or exactly one neighbour.
\end{Lemma}

\begin{proof}
	Assume that
	$A=\{\alpha+\frac{1}{\alpha}, \omega\alpha+\frac{1}{\omega\alpha},
	\omega^2\alpha+\frac{1}{\omega^2\alpha}\}$
	for some $\alpha\in\gf_{2^n}\setminus\gf_4$.
	For any $0\leq i\leq 2$,
	there exists a neighbour of $A$ {which is adjacent} at $\omega^i\alpha+\frac{1}{\omega^i\alpha}$
	if and only if $$\tr\left(\frac{1}{\omega^i\alpha+\frac{1}{\omega^i\alpha}+1}\right)=0.$$
	Further, {by using the property that $1+\omega+\omega^2=0$ and $\omega^3=1$}, we have
	{
	\begin{equation}
	\label{traceeq}
	\begin{array}{lll}
		\tr\Big(\frac{1}{\omega^i\alpha+\frac{1}{\omega^i\alpha}+1}\Big)
		  &=& \tr\left(\frac{\omega^i\alpha}{(\omega^i\alpha)^2+(\omega^i\alpha)+1}\right) \\[2ex]
		  &=& \tr\left(\frac{\omega^i\alpha}{(\omega^i\alpha)^2+(\omega^i\alpha)(\omega+\omega^2)+\omega\cdot\omega^2}\right) \\[2ex]
		  &=& \tr\left(\frac{\omega^i\alpha}{(\omega^i\alpha +\omega )(\omega^i\alpha+\omega^2)}\right)  \\[2ex]
                  &=& \tr\left(\frac{\omega^i\alpha}{\omega^i\alpha +\omega }\right) +
			\tr\left(\frac{\omega^i\alpha}{\omega^i\alpha +\omega^2 }\right) \\[2ex]
                  &=& \tr\left(\frac{1}{1+\omega^{1-i}\alpha^{-1}}\right) + \tr\left(\frac{1}{1+\omega^{2-i}\alpha^{-1}}\right).
	\end{array}
	\end{equation}
	It is easy to verify that}
	\begin{equation*}
		\sum_{i=0}^2\left(\tr\left(\frac{1}{1+\omega^{1-i}\alpha^{-1}}\right) +
		\tr\left(\frac{1}{1+\omega^{2-i}\alpha^{-1}}\right)\right)=0.
        \end{equation*}
	{Therefore by Equation (\ref{traceeq}) we know that
        $$ 0=\sum_{i=0}^2\tr\Big(\frac{1}{\omega^i\alpha+\frac{1}{\omega^i\alpha}+1}\Big), $$
        which follows that}
	$\tr\left(\frac{1}{\omega^i\alpha+\frac{1}{\omega^i\alpha}+1}\right), 0\leq i\leq 2$
	are either all zeros or exactly one zero.
	Hence  $A$ has either three neighbours or exactly one neighbour.
$\hfill\Box$\end{proof}

\begin{definition}
	\label{ABType}
	Let $A$ be a TS. If $A$ has three neighbours, then we call $A$ a fat TS.
	Otherwise, we call it a slim TS.
\end{definition}

By Lemma \ref{Neigh}, a TS is either fat or slim. A TS $A_{\alpha}$ is a
fat TS if and only if
$\tr(\frac{1}{1+\alpha^{-1}}) = \tr(\frac{1}{1+\omega\alpha^{-1}})=\tr(\frac{1}{1+\omega^2\alpha^{-1}})$,
and if and only if all the elements of $A_{\alpha}$ are in $T_3$. For a slim TS,
two of its elements are in $T_2$ while exactly one element is in $T_3$.

\begin{definition}\label{def_G}
Define the graph $\mathcal{G}=(V, E)$ from TSs as follows:
\begin{enumerate}
  \item[-] The vertices are all TSs $A_{\alpha}$, where $\alpha\in\gf_{2^n}\setminus\gf_4$;
  \item[-] Two vertices $A_{\alpha}, A_{\beta}$ are joined by an edge if and only if they are adjacent.
\end{enumerate}
\end{definition}

We summarize some properties of the graph $\mathcal{G}$ defined above in the following result.
\begin{proposition}
  Let $\mathcal{G}$ be the graph defined above. Then the following results hold:
  \begin{enumerate}[(i)]
     \item The degree of each vertex is either $1$ or $3$;
     \item The rank of $M$ is $\frac{2^{n+1}-5}{3}$ if and only if the graph
       $\mathcal{G}$ does not have a $3$-regular subgraph,
       {i.e. all vertices in the subgraph have degree exactly $3$}.
  \end{enumerate}
\end{proposition}

\begin{proof}
  (i) The result follows from Lemma \ref{Neigh}.

  (ii) Note that any  row vector of $M$ has Hamming weight either $2$ or $3$.
  For a row vector with Hamming weight $3$, its support corresponds to a TS; while
  for a row vector with Hamming weight $2$, its support corresponds to
  a set with form $\{\beta, 1+\beta\}, \beta\in \gf_{2^n}\setminus\gf_2$.
  Assume that the row vectors of $M$ are linear dependent over $\gf_2$.
  Since the supports of any two vectors of $M$ with the same
  Hamming weight are disjoint, we have $\xi_{1}+\cdots+\xi_{s}=\eta_{1}+\cdots+\eta_t$, where $\xi_i$ are vectors
  with Hamming weight $3$ and $\eta_j$ are those with weight $2$.
  Let the corresponding TS of $\xi_i$ be $A_i$ and let  $A_1=\{a,b,c\}$.
  Then $a$ and $1+a$ are in $\bigcup_{j=1}^t \supp(\eta_j)=\bigcup_{j=1}^s \supp(\xi_j)$.
  Hence $1+a$ is in the support of  $\xi_i$
  for some $2\leq i\leq s$. So
  are $1+b$ and $1+c$. It then follows that $A_1$ has $3$ neighbours and the set of its neighbours
  is a subset of $\{A_i| 2\leq i\leq s\}$.
  Denote by $\mathcal{H}$ the subgraph of $\mathcal{G}$ formed by the vertices $A_i, 1\leq i\leq s$.
  Then $\mathcal{H}$ is a $3$-regular subgraph of $\mathcal{G}$.
  We complete the proof.
$\hfill\Box$\end{proof}

We use the following table to list some properties of the graph $\mathcal{G}$ defined above.
For the definition of the girth, connected components and diameter, please refer to
any textbook on graph theory. The value of diameter in the table refers to the
largest diameter of each connected components.

\begin{center}
\label{tab:table3}
\tabcaption{Computational results of the graph $\mathcal{G}$ on $\gf_{2^n}$
for $6\leq n\leq 12$ with $n$ even}
 \begin{tabular}{|c|c|c|c|c|c|} \hline
   $n$  &$\#$ of vertices &$\#$ of edges & Girth           & $\begin{array}{ll}
				 & \#\ \mbox{of connected} \\
				 & \mbox{components}
			     \end{array} $  & Diameter  \\ \hline
  $6$   &$10$ & $6$ & no cycle        & $4$                          & $2$     \\ \hline
  $8$   &$42$ & $35$ & $8$             & $8$      & $6$  \\ \hline
  $10$  &$170$   & $120$   & $5$             & $51$     & $4$ \\ \hline
  $12$  &$682$ & $517$ & $9$         & $170$       &$18$  \\ \hline
\end{tabular}
\end{center}

\section{Non-decomposable Preferred Boolean Functions}

It is known from Proposition \ref{PFBrela} that the set $\mathcal{PBF}$ is a $\gf_2$-subspace.
To obtain linear independent PBFs, we focus on non-decomposable PBFs in this section.
After introducing its definition, we give a characterization of non-decomposable PBFs.
Then a large subspace of PBFs is explicitly constructed, which can lead to
many differentially $4$-uniform permutations.

\begin{definition}
	Let $f$ be a nonzero PBF. If there exist two PBFs $f_1$ and $f_2$ such that $f=f_1+f_2$
	and $\supp(f_i)\subsetneq \supp(f), 1\leq i\leq 2$,
	then $f$ is called \textit{decomposable}. Otherwise it is called \textit{non-decomposable}.
\end{definition}

Before giving the characterization of non-decomposable PBFs, we first state
some properties of them.
Let $f$  be a non-decomposable PBF with $f(0)=0$.  By Theorem \ref{main1}(ii), we have
$Mf^{\rm{T}}=0$.
Since $f(x+1)=f(x)$ holds for any $x\in \gf_{2^n}$,
the weight of $f$ must be even.
Assume that $|\supp(f)|=2t$ and $\supp(f)=\{\beta_i, \beta_i+1 | 1\leq i\leq t\}$
for some positive integer $t$.
Let $A$ be a TS. Then $|\supp(f)\cap A|=0$ or $2$.
In the following, assume that there are $r$ ($0\leq r\leq t$) TSs
$A_{i}=\{a_i, b_i, a_i+b_i\}$ such that $\supp(f)\cap A_i=\{a_i, b_i\}$.
Since $\mathcal{TS}$ is a disjoint union of the set $S=\{x\in\gf_{2^n}\setminus\gf_2|\tr(1/x)=0\}$,
we have $\supp(f)\cap S=\bigcup_{i=1}^{r}(A_i\cap \supp(f))$. \\

Now we present the main theorem in this section. Recall that the sets
$T_i $ ($1\leq i\leq 3$) are defined in \eqref{eqT}.

\begin{theorem}
  \label{NonDecSolution}
  Let $f$  be a Boolean function with $n$ variables.
{
  Assume that $|\supp(f)|=2t$ and there are $r\; (0\leq r\leq t)$ TSs
  $A_{i}=\{a_i, b_i, a_i+b_i\}$ such that $\supp(f)\cap A_i=\{a_i, b_i\}$.
  }
	Then the following results hold:
	\begin{enumerate}[(i)]
	  \item If $t=1$, then $f$ is a non-decomposable PBF if and only if
	    $r=0$ and there exists $\beta \in T_1$ such that $\supp(f)=\{\beta, 1+\beta\}$;
	  \item If $t=2$, then $f$ is a non-decomposable PBF if and only if
	    $r=1$ and there exists a slim
		  TS $A=\{\beta_1, \beta_2, \beta_1+\beta_2\}$ such that
		  $\supp(f)=\{\beta_1, \beta_2, 1+\beta_1, 1+\beta_2\}$, where
		  $\beta_1, \beta_2\in T_2$;
	  \item If $t\geq 3$, then either $r=t$ or $r=t-1$.
	    Furthermore,
            \begin{enumerate}[(a)]
	      \item
	    If $r=t$, then $f$ is a non-decomposable PBF if and only if
		there exist fat TSs $A_1=\{\beta_1, \beta_2, \beta_1+\beta_2\}$,
		  $A_i=\{1+\beta_{i-1}, \beta_{i+1}, 1+\beta_{i-1}+\beta_{i+1}\}$,
		  $2\leq i\leq t-1$, and $A_t=\{1+\beta_{t-1}, 1+\beta_{t}, \beta_{t-1}+\beta_{t}\}$ such that
		  $A_1, \cdots, A_{t-1}$ and $A_t$ form  a circle of TSs,
		  and $\supp(f)=\{\beta_i, 1+\beta_i| 1\leq i\leq t\}$.
	      \item	
            If $r=t-1$,  then $f$ is a non-decomposable PBF if and only if
	    there exist TSs $A_1=\{\beta_1, \beta_2, \beta_1+\beta_2\}$,
		  $A_2=\{1+\beta_1, \beta_3, 1+\beta_1+\beta_3\}$, and
		  $A_i=\{1+\beta_{i}, \beta_{i+1}, 1+\beta_{i}+\beta_{i+1}\}$,
		  $3\leq i\leq r$ such that $A_1, A_r$ are slim TSs and
		  $A_2, \cdots, A_{r-1}$ are fat TSs,
		  and $\supp(f)=\{\beta_i, 1+\beta_i|1\leq i\leq t\}$.
	    \end{enumerate}
	\end{enumerate}
\end{theorem}

\begin{proof}
	(i) Let $t=1$ and assume $\supp(f)=\{\beta, 1+\beta\}$. Then $\supp(f)\cap S=\emptyset$. Otherwise,
	there exists a TS $A$ such that $\supp(f)\subseteq A$. Then it follows that $1=\beta+(1+\beta)\in A$,
	which is a contradiction as $1\not\in S$. Hence $\beta\in T_1 $ and $r=0$.
	Conversely, for any $\beta\in T_1$,
	let $g$ be a vector with $\supp(g)=\{\beta, 1+\beta\}$. Then clearly
	$g$ is a non-decomposable PBF {by Theorem \ref{ChPBF}}.

	(ii) First note that, if $t\geq 2$ and $f$ is non-decomposable,
	we may see from (i) that $\beta_i\notin T_1$
	holds for any $1\leq i\leq t$ ({as otherwise assume $\beta_i\in T_1$, then
	$f$ can be decomposed into the sum
	of two PBFs, one with support set $\supp(f)\setminus\{\beta_i, \beta_i+1\}$
	and one with support set $\{\beta_i,\beta_i+1\}$}).
	In other words, for any $1\leq i\leq t$, either $\beta_i\in T_2$ or
	$\beta_i\in T_3$.
	Then $\{\beta_i, 1+\beta_i\}\cap S\neq \emptyset$
	for any $1\leq i\leq t$. Hence $r\geq \lceil\frac{t}{2}\rceil\geq 1$.

	Now, let $t=2$. Assume that $f$ is a non-decomposable PBF with support set
	$\{\beta_1,\beta_1+1,\beta_2,\beta_2+1\}$.
	Let $A_1$ be a TS such that $|\supp(f)\cap A_1|=2$. Clearly, for any $i=1,2$,
		at most one of $\beta_i$ and $ 1+\beta_i$ is in $A_1$. Without loss of generality,
	we assume that $\beta_1, \beta_2\in A_1$. Hence $A_1=\{\beta_1,\beta_2,\beta_1+\beta_2\}$.
	Now we claim that $\beta_1\in T_2$. Otherwise, we have
	$\beta_1\in T_3$, which means that $1+\beta_1\in S\cap \supp(f)$.
	Since $1+\beta_1\notin A_1$, let $A_2$ be the TS
	such that $1+\beta_1\in A_2$. Then $|A_2\cap \supp(f)|=2$, which deduces that $1+\beta_2\in A_2$. Hence
	$(1+\beta_1)+(1+\beta_2)=\beta_1+\beta_2\in A_2$, which  contradicts the fact
	that  different TSs are disjoint. Thus $\beta_1\in T_2$. Similarly, we can show that
	$\beta_2\in T_2$. Therefore, {$A_1=\{\beta_1,\beta_2,\beta_1+\beta_2\}$} is a slim TS
	and ${\supp(f)=\{\beta_1,\beta_1+1,\beta_2,\beta_2+1\}}$, which implies that
	{$r=1$}.

	{Conversely}, let $A=\{a, b, a+b\}$ be a slim TS and $a, b\in T_2$.
	Then the vector $v_g$ satisfying $\supp(v_g)=\{a, b, 1+a, 1+b\}$ is a non-decomposable
	PBF.

        (iii)
	From now on, we assume that $t\geq 3$. Then $r\geq \lceil\frac{t}{2}\rceil\geq 2$
	if $f$ is non-decomposable. First assume that $f$ is non-decomposable,
	we distinguish the following two cases:

	(a) All TSs $A_i$ with $A_i\cap\supp(f)\neq\emptyset$ are fat TSs.
	By the notations at the beginning of this section, we assume that there are $r$
	such TSs, where $2\leq r\leq t$;

	(b) {There exists at least one slim TS $A_i$ among those $A_i\cap\supp(f)\neq\emptyset$.}
	Further, for any slim TS $A_i=\{a_i, b_i, a_i+b_i\}$, $a_i, b_i \in \supp(f)$,
	we have $a_i\notin T_2$ or $b_i\notin T_2$, or equivalently, $a_i\in T_3$ or $b_i\in T_3$.
        {Otherwise, if $a_i,b_i\in T_2$, we may decompose $f$ into the sum of $f_1$ and $f_2$,
	where $f_1,f_2$ are PBFs, one with support set $\supp(f)\setminus\{a_i,b_i,1+a_i,1+b_i\}$
	and one with support set $\{a_i,b_i,a_i+1,b_i+1\}$ (by (ii) such a function is a PBF).}

	{Case (a)}: Assume that $t\geq 3$ and all $A_i$ are fat TSs for $1\leq i\leq r$.
	Assume that $A_1=\{\beta_1, \beta_2, \beta_1+\beta_2\}$.
	Since $A_1$ is fat, we have $1+\beta_1, 1+\beta_2\in S\cap \supp(f)$.
	Clearly, $1+\beta_1$ and $1+\beta_2$ can not be in {one} TS.
	Assume that $1+\beta_1\in A_2$ and $1+\beta_2\in A_3$.
	Denote by $\beta_3$ the other element in $A_2\cap \supp(f)$.
	Then $A_2=\{1+\beta_1, \beta_3, 1+\beta_1+\beta_3\}$.
	Similarly, we have $1+\beta_3\in S\cap \supp(f)$ since $A_2$ is fat.

	If $t=3$, then $r\leq 3$. Hence $A_3=\{1+\beta_2, 1+\beta_3, \beta_2+\beta_3\}$.
	Then $\supp(f)=\{\beta_i, 1+\beta_i, 1\leq i\leq 3\}$ is a non-decomposable PBF.
	If $t>3$, then $1+\beta_3\notin A_3$ since $f$ is non-decomposable. Without loss of generality,
	assume that $1+\beta_3\in A_4$ and $\beta_4\in A_3$.
	So on and so forth, since $f$ is non-decomposable, we have $r=t$
	and
	$A_1=\{\beta_1, \beta_2, \beta_1+\beta_2\}$,
	$A_i=\{1+\beta_{i-1}, \beta_{i+1}, 1+\beta_{i-1}+\beta_{i+1}\}$,
	  $2\leq i\leq t-1$, $A_t=\{1+\beta_{t-1}, 1+\beta_{t}, \beta_{t-1}+\beta_{t}\}$.
	Then   $A_1, \cdots, A_t$ forms a circle of fat TSs with length $t$.

	{Case (b):} Assume that $t\geq 3$ and there exists $1\leq i\leq r$ such that $A_i$ is a slim TS.
	Without loss of generality, we assume that $A_1=\{\beta_1, \beta_2, \beta_1+\beta_2\}$ is a slim TS.
	Then $\beta_1\in T_3$ or $\beta_2\in T_3$.
	Without loss of generality, we assume that $\beta_1\in T_3$ and $\beta_2, \beta_1+\beta_2\in T_2$
	(Note that there is one and only one element of an slim TS in $T_3$).
	It follows that $1+\beta_1\in S\cap \supp(f)$ and $1+\beta_2\notin S$. Without loss of generality,
	assume that $1+\beta_1\in A_2$ and $\beta_3\in A_2$. Then $A_2=\{1+\beta_1, \beta_3, 1+\beta_1+\beta_3\}$.

	If $A_2$ is slim, then $\beta_3\in T_2$ and $1+\beta_3\notin S$.
	Hence the vector $v$ such that $\supp(v)=\{\beta_i, 1+\beta_i| 1\leq i\leq 3\}$ is a PBF.
	Since $V_f$ is non-decomposable, this can only happen when $t=3$. And in this case, we have $r=2=t-1$.

	If $A_2$ is fat, then $\beta_3\in T_3$ and $1+\beta_3\in S\cap \supp(f)$.
	Similarly, we can assume that
	$1+\beta_3\in A_3$ and $\beta_4\in A_3$.  Hence $A_3=\{1+\beta_3, \beta_4, 1+\beta_3+\beta_4\}$.
	Similarly as just discussed, if  $A_3$ is slim, then $t=4$, $r=3=t-1$, and
	$\supp(f)=\{\beta_i, 1+\beta_i| 1\leq i\leq 4\}$ is a non-decomposable PBF.
	If $A_3$ is fat, then we can assume that $A_4=\{1+\beta_4, \beta_5, 1+\beta_4+\beta_5\}$.
	So on and so forth, since $f$ is non-decomposable, we have $r=t-1$.
	Further, $A_1=\{\beta_1, \beta_2, \beta_1+\beta_2\}$,
	$A_2=\{1+\beta_1, \beta_3, 1+\beta_1+\beta_3\}$,
	and $A_i=\{1+\beta_{i}, \beta_{i+1}, 1+\beta_{i}+\beta_{i+1}\}$,
	$3\leq i\leq r$, where
	$A_1, A_r$ are slim TSs and $A_2, \cdots, A_{r-1}$ are fat TSs.

	The proof of the converse part is not difficult and we omit it here.
	We complete the proof.
$\hfill\Box$\end{proof}

The following proposition follows directly from Theorem \ref{NonDecSolution}.

\begin{proposition}
	Let $n=2k$  be an even integer and $\mathcal{G}$ be the graph defined in Definition \ref{def_G}.
	Then the following results hold:
	\begin{enumerate}[(i)]
	  \item The number of type (i) non-decomposable PBFs in Theorem \ref{NonDecSolution}
	    is
	    $$ \frac{1}{8}\left( 2^n + 1 - 2^{n/2+1}\cos\left(n\arccos(1/\sqrt{8})\right)\right), $$
 i.e. half of the cardinality of $T_1$;	
	  \item The number of type (ii) non-decomposable PBFs is the number of the slim TSs;
	  \item A type (iii)(a) non-decomposable PBF with weight $2t$ exists
	    if and only if there exists a cycle in $\mathcal{A}$ of length $t$,
	    where $\mathcal{A}$ is the subgraph of $\mathcal{G}$ generated by all fat vertices;
	  \item A type (iii)(b) non-decomposable PBF with weight $2t$ exists
	    if and only if there exists a path of  $\mathcal{G}$ with length $t$,
	    where the starting and ending vertices are slim vertices, and the others are fat.
	\end{enumerate}
\end{proposition}

Computer experiments on small fields suggest that there exist many non-decomposable PBFs of the type (iii)(b)
and much few those of the type (iii)(a). In the following, we give some experiment results about the
non-decomposable PBFs of the type (iii)(a).
Note that the minimal value of $t$ such that a type (iii)(a) non-decomposable PBF exists
is the girth of the subgraph of $\mathcal{G}$ generated by all fat TSs.
Also, the maximal value of such $t$ is $2d+1$, where $d$ is the maximal diameter
of all connected components of this subgraph.
We use the following table to list the properties of the subgraph of $\mathcal{G}$
generated by the fat TSs. The notations are the same as Table 3.

 \begin{center}
    \tabcaption{Computational results of the subgraph of $\mathcal{G}$
    generated by the fat TSs on $\gf_{2^n}$
    for $6\leq n\leq 12$ with $n$ even}
    \begin{tabular}{|c|c|c|c|c|c|} \hline
	   $n$   & $\#$ of vertices &$\#$ of edges & Girth           & $\begin{array}{ll}
	                                 & \#\ \mbox{of connected} \\
					 & \mbox{components}
			             \end{array} $  & Diameter  \\ \hline
	  $6$   & $1$ & $0$  & no cycle        & $1$     & $0$     \\ \hline
	  $8$   & $14$ & $13$ & $8$             & $2$      & $4$  \\ \hline
	  $10$  & $35$  & $15$ & $5$             & $21$     & $2$ \\ \hline
	  $12$  & $176$ & $138$ & $9$         & $43$      & $18$  \\ \hline
   \end{tabular}
  \end{center}

Finally, we apply Theorem \ref{NonDecSolution} to explicitly construct a large set of
PBFs, and hence obtain many differentially $4$-uniform permutations.

We first introduce some notations.
For any $\beta\in T_1$, define a function $f_\beta$ as
$$ f_{\beta}(x)=(x+\beta)^{2^n-1}+(x+\beta+1)^{2^n-1}. $$
Let $A=\{a_1, a_2, a_1+a_2\}$ be a slim TS, where $a_1, a_2\in T_2$.
Define the function $f_A$ as
$$f_{A}(x)=\sum_{i=1}^2\left((x+a_i)^{2^n-1}+(x+a_i+1)^{2^n-1}\right).$$

\begin{theorem}
  \label{LowWeight}
Let $\mathcal{X}$ be the set of non-decomposable PBFs with weight $2$ or $4$.
Then
\begin{equation*}
  \mathcal{X}=\Big\{f_{\beta}, f_A \;|\; \beta\in T_1, A \text{ is a slim TS}\Big\}
\end{equation*}
and
$|\mathcal{X}|= \frac{|T_1|}{2}+\Big|\{ A\;|\; A \text{ is a slim TS}\}\Big|$,
where $f_{\beta}, f_A$ are as defined above.
Define the matrix $X$ with each row the value vector of
a Boolean function in $\mathcal{X}$. Then the rank of $X$ is $|\mathcal{X}|$.
Therefore, the rows of $X$ generate a subspace of $\mathcal{PBF}$
with dimension $|\mathcal{X}|$. Denote this subspace by  $\mathcal{PBF}_4$.
By Corollary \ref{th_general}, for each $f\in\mathcal{PBF}_4$,
we can construct a differentially $4$-uniform permutation
on $\gf_{2^n}$. Therefore we explicitly obtain $2^{|\mathcal{X}|}$ such functions.
\end{theorem}
\begin{proof}
The proof is simple and we omit it here.
$\hfill\Box$\end{proof}

We use the following table to list $\dim{(\mathcal{PBF}_4)}$,
the dimension of the subspace of differentially $4$-uniform
permutations obtained in Theorem \ref{LowWeight} for $6\leq n\leq 14$.
It seems that $\dim{(\mathcal{PBF}_4)}=|\mathcal{X}|=2^{n-2}$. This hints that
the dimension of $\mathcal{PBF}_4$ is about $\frac{3}{4}$ of that of the whole space $\mathcal{PBF}$.

\begin{center}
   \tabcaption{Dimension of the differentially $4$-uniform
     permutations on $\gf_{2^n}$ obtained in Thoerem \ref{LowWeight},
     $6\leq n\leq 14$ and $n$ even}
       \begin{tabular}{|c|c|c|c|} \hline
       \label{table-thm1}
       $n$    & $\dim(\mathcal{PBF}_4)$ (Thm \ref{LowWeight})   &  $\dim(\mathcal{PBF})$    \\ \hline
	$6$    & $16$   & $22$      \\ \hline
	$8$    & $64$   & $86$        \\ \hline
	$10$    & $256$   & $342$        \\ \hline
	$12$    & $1024$   & $1366$        \\ \hline
	$14$    & $4096$   & $5462$        \\ \hline
 \end{tabular}
\end{center}

\section{Conclusions}
In this paper we propose a particular type of Boolean functions, preferred Boolean functions,
to characterize the preferred functions. This enables us to give a more efficient method
to construct new differentially $4$-uniform permutations over $\gf_{2^{2k}}$.
Furthermore, it is proven that such Boolean functions can be
determined by solving linear equations. Hence the number of them can be determined
by computing the rank of the coefficient matrix. As an application, we show that
the number of CCZ-inequivalent differentially $4$-uniform permutations
over $\gf_{2^n}$ ($n$ even) is at least
$2^{\frac{2^n+2}{3}-4n^2-2n}$,  which implies the number of the CCZ-inequivalent classes
of such permutations grow exponentially when $n$ grows. This positively answer an open
problem proposed in \cite{QTTL}. Finally, we study the non-decomposable preferred Boolean
functions, and use them to construct more differentially $4$-uniform permutations explicitly.
The obtained functions in this paper may provide more choices for the
design of Substitution boxes.

At last, we make a remark on the {number of} permutations on $\gf_{2^n}$
($n$ odd) with {differential uniformity at most $4$ and with algebraic degree $n-1$}.
Recall that when $n $ is odd, {the inverse function {$x^{-1}$} is an APN permutation on $\gf_{2^n}$}.
It is clear that the differential uniformity of $G(x)=x^{-1}+f(x)$
is at most $4$ and with algebraic degree $n-1$, where $f\in\mathcal{BF}_n$ is any Boolean function.
By \cite[Result 2]{QTTL} (see proof in \cite[Lemma1]{4uni-seta}),
there exist $2^{2^{n-1}}$ Boolean functions $f$ such that $G(x)=x^{-1}+f(x)$ is a permutation.
As what we show in the proof of Theorem \ref{main1}(iii),
for any $(n,n)$-function, there are at most
$(2^n)^{4n+2}=2^{4n^2+2n}$ functions which are CCZ-equivalent to it.
Therefore we have at least
$$
\frac{2^{2^{n-1}}}{2^{4n^2+2n}} = 2^{2^{n-1}-4n^2-2n}
$$
CCZ-inequivalent permutations over $\gf_{2^n}$ ($n$ odd) with differential uniformity at most 4.
It is easy to see that such number grows exponentially when $n$ increases.

\section*{Acknowledgement}
{We would like to thank the anonymous reviewer for the valuable comments which significantly improve
the quality and presentation of this paper.}

Part of this work was done when the first author visited Hong Kong University of Science and Technology.
He would like to thank Prof. Cunsheng Ding for the kind hospitality and discussions during this period.

The research of L. Qu is supported by the National Natural Science Foundation of China (No. 61272484),
the Research Project of National University of Defense Technology under Grant CJ 13-02-01 and
the Program for New Century Excellent Talents in University (NCET).
The research of C. Li is supported by the National Basic Research Program of China 2013CB338002 and the
open research fund of Science and Technology on Information Assurance Laboratory (Grant No. KJ-12-02).

\end{document}